\documentclass[conference]{IEEEtran}
\IEEEoverridecommandlockouts

\ifCLASSINFOpdf
   \usepackage[pdftex]{graphicx}
   \DeclareGraphicsExtensions{.pdf,.jpeg,.png}
\else
\fi
\usepackage[cmex10]{amsmath}
\usepackage{algorithmic}
\usepackage{graphicx,amsmath,amssymb,hyperref}
\usepackage[T1]{fontenc}
\newtheorem{theorem}{Theorem}

\newtheorem{lemma}{Lemma}

\newtheorem{example}{Example}

\newcommand{\beqno}{ \begin{equation*} }
\newcommand{\eeqno}{ \end{equation*} }
\newcommand{\beq}{ \begin{equation} }
\newcommand{\eeq}{ \end{equation} }

\newcommand{\calB}{\mathcal{B}}
\newcommand{\calD}{\mathcal{D}}

\newcommand{\Mod}[1]{\ (\text{mod}\ #1)}

\begin{document}
%
\title{Capacity of Sum-networks for Different Message Alphabets}

\author{\IEEEauthorblockN{Ardhendu Tripathy and Aditya Ramamoorthy}
\IEEEauthorblockA{Department of Electrical and Computer Engineering,
Iowa State University, Ames, Iowa 50011\\
Email: \{ardhendu,adityar\}@iastate.edu}
\thanks{This work was supported in part by NSF grants CCF-1320416 and CCF-1149860.}
}


%


\maketitle

\begin{abstract}
A \textit{sum-network} is a directed acyclic network in which all terminal nodes demand the `sum' of the independent information observed at the source nodes. Many characteristics of the well-studied multiple-unicast network communication problem also hold for sum-networks due to a known reduction between instances of these two problems. Our main result is that unlike a multiple unicast network, the coding capacity of a sum-network is dependent on the message alphabet. We demonstrate this using a construction procedure and show that the choice of a message alphabet can reduce the coding capacity of a sum-network from $1$ to close to $0$.

\end{abstract}


%
\IEEEpeerreviewmaketitle

\section{Introduction}
Function computation using network coding is an extensively studied research area \cite{appuFKZ11, appuFKZ13,ramamoorthyL13, raiD12}.  The typical setting is one of directed acyclic networks with error-free links, a set of source nodes that generate independent and identically distributed information and a set of terminal nodes with specific demands. The work of \cite{appuFKZ11, appuFKZ13} considers general functions, but networks with only one terminal. A different line of work considers networks with multiple terminals that each need a simple function such as the sum \cite{ramamoorthyL13, raiD12}.

In this work, we consider sum-networks. 
There are several results about the solvability of multiple unicast networks, many of which also hold for sum-networks by virtue of a reduction from a multiple unicast instance to a sum-network instance as described in \cite{raiD12}. 
Specifically, it was shown in \cite{DoughertyFZ05} that linear network coding in the most general sense is unable to attain the coding capacity of multiple unicast networks. 
Reference \cite{DoughertyFZ08} showed that there exists a multiple unicast network which is scalar linear solvable over a finite field $\mathcal{F}$ if and only if a corresponding collection of integer-coefficient polynomials have a common root in $\mathcal{F}$. Analogous statements for sum-networks were noted to be true in \cite{raiD12} because of the equivalence. Reference \cite{CannonsDFZ06} showed that the coding capacity of a general multiple unicast network is independent of the alphabet used.

A network code for a sum-network is said to have rate $m/n$ if in $n$ time slots, one can multicast the sum $m$ times to all the terminals. A network is called solvable if it has a $(m,m)$ code and not solvable otherwise. Prior work has investigated the effect of field characteristic on sum-network capacity. Specifically, \cite{raiD12} constructs sum-networks whose solvability depends on the field characteristic under the restriction of linear network coding. Given a ratio $p/q<1$, it was shown in \cite{raiD13} that there exist sum-networks which have $p/q$ as their coding capacity. Reference \cite{tripathyR14} described a systematic construction procedure for sum-networks with capacity $p/q$; these are typically smaller than those in \cite{raiD13}. The capacity of these sum-networks however was independent of the message alphabet.

\subsection{Main contribution}
\noindent $\bullet~$ In this work, we construct sum-networks whose capacity depends on the characteristic of the finite field $\mathcal{F}$ chosen as the message alphabet. As a specific example, we construct sum-networks for which the capacity is $1$ if $ch(\mathcal{F})=2$ and can be made arbitrarily close to  $0$ if $ch(\mathcal{F})\neq 2$.

The problem is formally posed in section \ref{sec:problem}, our construction for the class of sum-networks is explained in section \ref{sec:construction}. A cut-set bound on the coding capacity is described in section \ref{sec:upbound} and a linear network code that attains this rate is described in section \ref{sec:lincode}. 
\section{Problem formulation}
\label{sec:problem}
We consider communication over a directed acyclic graph (DAG) $G=(V,E)$ where $V$ is the  set of nodes and $E \in V \times V$ are the edges denoting the delay-free communication links between them.
Subset $S \subset V$ denotes the source nodes and $T \subset V$ denotes the terminal nodes. The source nodes have no incoming edges and the terminal nodes have no outgoing edges. Each source node $s_i \in S$ generates an independent random process $X_i$, such that the sequence of random variables $X_{i1}, X_{i2}, \dots$ indexed by time are i.i.d. and each $X_{ij}$ takes values that are uniformly distributed over a finite alphabet $\mathcal{F}$ that is assumed to be a finite field such that $|\mathcal{F}| = q$; the characteristic of $\mathcal{F}$ will be denoted by $ch(\mathcal{F})$. An edge starting at $u$ and ending in $v$ in the DAG will be denoted as the ordered pair $(u,v)$. We define head$(u,v)=v$ and tail$(u,v)=u$. Each edge is of unit capacity and can transmit one symbol from $\mathcal{F}$ per unit time.
We use the notation $\text{In}(v)$ to represent the set of incoming edges at node $v \in V$. If $e=(u,v)$, then we set $\text{In}(e)=\text{In}(u)$.
%

A network code is an assignment of encoding functions (we call these the local encoding functions) to each edge in $E$ and a decoding function to each terminal in $T$. The local encoding function for an edge connected to a set of sources depends only those particular source values. Likewise, the local encoding function for an edge that is not connected to any source depends on the values received on its incoming edges and the decoding function for a terminal depends only on its incoming edges. As we consider directed acyclic networks, it can be seen that we can also define a global encoding function that expresses the value transmitted on an edge in terms of the source values.
\begin{itemize}
\item Local encoding function for edge $e$. 
\begin{align*}
\tilde{\phi}_e &: \mathcal{F}^m \rightarrow \mathcal{F}^n ~~ \text{if tail}(e)\in S,  \\
\tilde{\phi}_e &: \mathcal{F}^{n|\text{In}(\text{tail}(e))|} \rightarrow \mathcal{F}^n ~~ \text{if tail}(e) \notin S.
\end{align*}

\item Decoding function for the terminal $t_i \in T$.
\begin{equation*}
\psi_{t_i} : \mathcal{F}^{n|\text{In}(t_i)|} \rightarrow \mathcal{F}^m
\end{equation*}
\end{itemize}

A network code is linear if all the edge and decoding functions are $\mathcal{F}$-linear. For the sum-networks that we consider, a $(m,n)$ fractional network code solution over $\mathcal{F}$ is such that the sum (over $\mathcal{F}$) of $m$ source symbols can be communicated to all the terminals in $n$ units of time. The rate of this network code is defined to be $m/n$. A network is said to be solvable if it has a $(m,m)$ network coding solution for some $m \geq 1$. A network is said to have a scalar solution if it has a $(1,1)$ solution. The supremum of all achievable rates is called the capacity of the network.

\section{Construction using BIBDs}
\label{sec:construction}
We begin by defining a $2-(v,k,\lambda)$ balanced incomplete block design (BIBD) \cite{Stinson}. It is a set system $\calD = (P,\calB)$ which has the following components.
\begin{itemize}
\item \textit{Points}: A set $P$ consisting of $v$ elements (called points), indexed in arbitrary order as $P = \{p_1,p_2,\ldots,p_v\}$.
\item \textit{Blocks}: A set $\mathcal{B}$ of size $b$ whose elements are $k-$subsets of $P$ such that $\mathcal{B}=\{B_1,B_2,\ldots,B_b\}$. $\calB$ satisfies the following \textit{regularity} property. For $p_i,p_j \in P, i \neq j$,
\begin{align*}
|\{B \in \mathcal{B} : p_i \in B, p_j \in B\}|=\lambda.
\end{align*}
\end{itemize}
For such a design, we can define an incidence matrix $A$ which is a $v \times b ~(0,1)$-matrix that records the incidence between points and blocks, i.e.,
\begin{align*}
A(i,j)=\left\{\begin{matrix}
1 &~\text{if}~p_i \in B_j, \\
0 &~\text{otherwise}.
\end{matrix}\right.
\end{align*}
It can be shown that each point is present in a fixed number of blocks (denoted by $r$). The following relations can be shown.
\begin{align*}
r&=\frac{\lambda (v-1)}{k-1}, \text{~and}\\
bk&=vr.
\end{align*} 
For any $p \in P$ and $B \in \mathcal{B}$, let
\begin{align*}
\langle p \rangle &=\{B \in \mathcal{B} : p \in B\}, \text{and}\\
\langle B \rangle&= \cup_{p \in B} \langle p \rangle=\cup_{p \in B}\{B' \in \mathcal{B} : p \in B'\} .
\end{align*} Thus we have $|\langle p\rangle|=r, ~\forall p \in P$. For the case of $\lambda =1$, we have that for any $p,p' \in B$, $\langle p\rangle \cap \langle p'\rangle=B$, i.e., a pair of points appear in a unique block.
\begin{example}
\label{ex:fano}
We describe the components of a $2-(7,3,1)$ design, which is also called a Fano plane. Letting numerals denote points and alphabets denote blocks for this design, we can write:
\begin{align*}
P&=\{1,2,3,4,5,6,7\},~\mathcal{B}=\{A,B,C,D,E,F,G\}, \\ A&=\{1,2,3\}, ~B=\{3,4,5\},~ C=\{1,5,6\},~D=\{1,4,7\}, \\  E&=\{2,5,7\}, ~F=\{3,6,7\},~G=\{2,4,6\}.
\end{align*}
The corresponding incidence matrix $A$ is shown below.
\begin{align}
A=
\begin{bmatrix}
1 & 0 & 1 & 1 & 0 & 0 & 0\\1 & 0 & 0 & 0 & 1 & 0 & 1\\1 & 1 & 0 & 0 & 0 & 1 & 0\\0 & 1 & 0 & 1 & 0 & 0 & 1\\0 & 1 & 1 & 0 & 1 & 0 & 0\\0 & 0 & 1 & 0 & 0 & 1 & 1\\0 & 0 & 0 & 1 & 1 & 1 & 0
\end{bmatrix}.
\label{eq:A}
\end{align}
\end{example}
We now construct a sum-network $G=(V,E)$ from any BIBD $\mathcal{D}$. 
We first describe the components of the vertex set $V$.
\begin{enumerate}
\item \textit{Source node set}: $S$ contains $v+b$ elements corresponding to points and blocks in $\calD$, i.e.,
\begin{align*}
S&=\{s_{p_1},s_{p_2},\ldots,s_{p_v},s_{B_1},s_{B_2},\ldots,s_{B_b}\}\\
&=\{s_p : p \in P\} \cup \{s_B : B \in \mathcal{B}\}.
\end{align*}
Each source node $s_{p_i} (\text{or}~s_{B_j})$ observes an independent unit-entropy random process $X_{p_i} (\text{or}~X_{B_j})$, respectively. The set of all source processes (or ``sources'') is $X=\{X_p : p \in P\} \cup \{X_B : B \in \mathcal{B}\}$.
\item \textit{Terminal node set}: $T$ also contains $v+b$ elements, in the same manner as $S$, i.e.,
\begin{align*}
T=\{t_p : p \in P\}\cup \{t_B : B \in \mathcal{B}\}.
\end{align*}
\item \textit{Intermediate nodes}: In addition to the above there are $2v$ vertices which are elements of $M^H \cup M^T$ (superscripts $H$ and $T$ denote head and tail, respectively) where
\begin{align*}
M^H=\{m^h_1, m^h_2, \ldots, m^h_v\}~\text{and} ~M^T=\{m^t_1,m^t_2, \ldots , m^t_v\}.
\end{align*}
\end{enumerate}
We set $V=S \cup T \cup M^H \cup M^T$ and thus $V$ contains $2(v+b)+2v$ vertices. 
The edge set $E$ of the directed acyclic network $G$ has the following components.
\begin{enumerate}
\item \textit{Bottleneck edges}: We introduce $v$ unit-capacity edges $e_i=(m^t_i,m^h_i),~ i \in \{1,2,\ldots v\}$. Thus the sets $M^H$ and $M^T$ denote the set of heads and tails of the bottleneck edges respectively. We also make the following connections for all $p_i \in P$
\begin{itemize}
\item $(s_{p_i}, m_i^t)$ and $(s_{B_j}, m_i^t)$ for all $B_j \in \langle p_i \rangle$,
\item $(m^h_i, t_{p_i})$ and $(m^h_i,t_{B_j})$ for all $B_j \in \langle p_i \rangle$.
\end{itemize}
Thus, each $m_i^t$ has $r+1$ incoming edges and each $m_i^h$ has $r+1$ outgoing edges. Denoting the set of edges introduced in this step as $M$, we get $|M|=v + 2v(r+1)$.
\item \textit{Direct edges}: The remaining edges are those that are not incident on either the tail or the head of the bottleneck edges and are referred to as the direct-edge set $D$.
It consists of the following unit-capacity edges for every $p_i \in P$ and $B_j \in \mathcal{B}$,
\begin{itemize}
\item $(s_{p_l},t_{p_i})$ for all $p_l \neq p_i$,
\item $(s_{B_l},t_{p_i})$ for all $B_l \notin \langle p_i \rangle$,
\item $(s_{p_l},t_{B_j})$ for all $p_l \notin B_j$, and
\item $(s_{B_l},t_{B_j})$ for all $B_l \notin \langle B_j \rangle$.
\end{itemize}
\end{enumerate}
We then set $E=M \cup D$. This completes the construction of the network $G$.
It can be verified now that each terminal is connected to every source node in the network by at least one path. To denote a general element of the set $E$ we will use the letter $e$.
Given a set of local encoding functions for each edge, denoted $\tilde{\phi}_e$, we can define the corresponding global encoding function for edge $e$, denoted $\phi_e$. In what follows, we will use the notation $\phi_e(X)$ to denote the $n$-length vector that is transmitted on edge $e$. For convenience, we define
\begin{align*}
\phi_{\text{In}\left(s_{p_i}\right)}(X) &= \begin{bmatrix}X_{p_i}^T & \mathbf{0}_{1 \times (n-m)}\end{bmatrix}^T, \text{~for all~}  p_i \in P\\
\phi_{\text{In}\left(s_{B_j}\right)}(X) &=\begin{bmatrix}X_{B_j}^T & \mathbf{0}_{1 \times (n-m)}\end{bmatrix}^T, \text{~for all~}  B_j \in \mathcal{B}\\
\phi_{\text{In}(v)}(X) &= \{\phi_e(X) : e \in \text{In}(v)\}, \text{~for all~} v \in V \setminus S\\
\phi_{\text{In}(e_i)}(X) &= \{\phi_e(X) : e \in \text{In}(m_i^t)\}.
\end{align*}
As will be apparent for our networks, nontrivial encoding functions will only be required on the bottleneck edges, $e_i, i = 1,\dots,v$. For brevity, we set $\phi_i(X) = \phi_{e_i}(X)$.

\begin{example}
\label{ex:fanograph}
Using the BIBD described in example \ref{ex:fano}, we construct a sum-network here. It has
\begin{figure}
\centering
\includegraphics[scale=0.6]{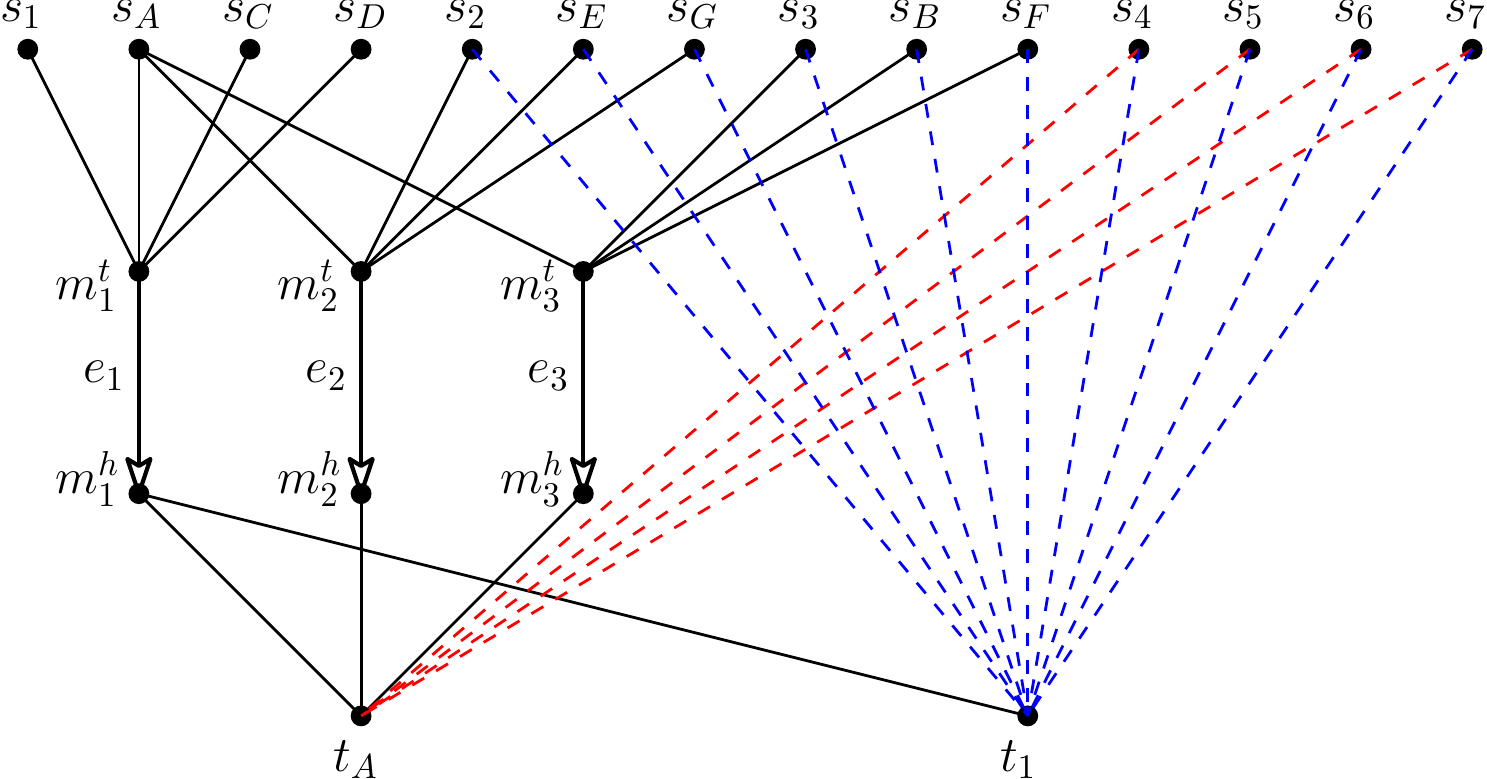}
\caption{Figure showing a part of the sum-network constructed using Fano plane. The $14$ source nodes are shown at the top of the figure and $2$ out of the $14$ terminal nodes are depicted at the bottom. All edges have unit capacity and point downward. The $3$ edges with the arrowheads ($e_1,e_2,e_3$) are the `bottleneck edges'. The red dashed lines are the `direct edges' for terminal node $t_A$ and the blue dashed lines are the `direct edges' for terminal $t_1$.}
\label{fig:sumnw}
\end{figure}
\begin{itemize}
\item fourteen source nodes $\{s_1,\ldots,s_7,s_A,\ldots,s_G\}$,
\item fourteen terminal nodes $\{t_1,\ldots,t_7,t_A,\ldots,t_G\}$,
\item fourteen intermediate nodes $\{m_1^h,\ldots,m_7^h,m_1^t,\ldots,m_7^t\}$, and
\item seven \hspace{-1.5mm} bottleneck \hspace{-1.5mm} edges \hspace{-1.5mm} $\{e_1\hspace{-0.5mm}=\hspace{-0.5mm}(m_1^t,m_1^h), \hdots, e_7\}$.
\end{itemize}
There will also be direct edges and edges connecting the bottlenecks to the sources and terminals.
Part of the constructed sum-network is shown in figure \ref{fig:sumnw}. We have shown the construction procedure on only three bottleneck edges, namely $e_1,e_2,e_3$ and two terminals $t_1, t_A$. The edges coming into and going out of other bottlenecks and terminals can be constructed similarly.
\end{example}

\section{Upper bound on coding capacity of constructed sum-network}
\label{sec:upbound}
Let $H(Y)$ be the entropy function for a random variable $Y$. We let $H(Y_1,Y_2,\ldots, Y_l)=H(\{Y_1,Y_2,\ldots,Y_l\})=H(\{Y_i\}_{1}^{l})$ for any $l > 1$. 
Suppose that there exists a $(m,n)$ fractional network code assignment, i.e. $\tilde{\phi}_e$ for $e \in E$ (and corresponding global encoding functions $\phi_e(X)$) and decoding functions $\psi_t$ for $t \in T$ so that all the terminals in $T$ can recover the sum of sources denoted by $Z=\sum_{p \in P} X_p+\sum_{ B \in \mathcal{B}} X_B$.
The following two lemmas demonstrate that certain partial sums can be computed by observing subsets of the bottleneck edges (see Appendix \ref{app:lemma1}, \ref{app:lemma2} for proofs).
\begin{lemma}
\label{lemma:ub1}
The value $X'_i=X_{p_i} + \sum_{ B \in \langle p_i \rangle} X_B$ can be obtained from $\phi_{i}(X)$ for all $i = 1, \dots, v$.
\end{lemma}
\begin{lemma}
\label{lemma:ub2}
The value of $\sum_{ p \in B_j} X_p  + \sum_{B \in \langle B_j \rangle} X_B$ can be computed from the set of values $\{\phi_{i}(X) : p_i \in B_j\}$ for all $B_j \in \mathcal{B}$.
\end{lemma}
\begin{theorem}
For a $2-(v,k,1)$ design $\mathcal{D}$, the coding capacity of the sum-network obtained using the given construction is at most $1$.
\label{thm:ub1}
\end{theorem} 
\begin{proof}
Let $X'_i=X_{p_i}+\sum_{ B \in \langle p_i \rangle} X_B$ and $\phi_{i}(X)$ is the function transmitted on the $i$-th bottleneck. Under a $(m,n)$ fractional network code, we have that
$H(X'_i)=m \log_2q ~\text{bits} ~\text{and}~ H(\phi_{i}(X))\leq n \log_2q ~\text{bits}$.
We then have
\begin{align}
&H(\{\phi_{i}(X)\}_{1}^{v})  \leq \sum_{i=1}^{v}H(\phi_{i}(X)) \leq vn \log_2q ~\text{bits, and} \nonumber \\
&H(\{\phi_{i}(X)\}_{1}^{v})=I\left(\{\phi_{i}(X)\}_{1}^{v} ; \{X'_i\}_{1}^{v}\right)\hspace{-0.5mm}+\hspace{-0.5mm}H\left(\{\phi_{i}(X)\}_{1}^{v}|\{X'_i\}_{1}^{v}\right) \nonumber \\
&\hspace{-0.5mm}=\hspace{-0.5mm}H\left(\{X'_i\}_{1}^{v}\right)\hspace{-0.5mm}-\hspace{-0.5mm}H\left(\{X'_i\}_{1}^{v}|\{\phi_{i}(X)\}_{1}^{v}\right)\hspace{-0.5mm}+\hspace{-0.5mm}H\left(\{\phi_{i}(X)\}_{1}^{v}|\{X'_i\}_{1}^{v}\right). \label{eq:th2}
\end{align}
We assume that all source random variables are independent and uniformly distributed over $\mathcal{F}^m$. That implies $X'_i$ are also uniform i.i.d. over $\mathcal{F}^m$ for all $i$ (see Appendix \ref{app:claim}).
Thus, we have that $H\left(\{X'_i\}_{1}^{v}\right)=vm\log_2q$ bits.
We expand the second term of \eqref{eq:th2} as
\vspace{-2mm}\begin{align*}
H\left(\{X'_i\}_{1}^{v}|\{\phi_{i}(X)\}_{1}^{v}\right)=\sum_{i=1}^{v}H\left( X'_i | \{X'_j\}_1^{i-1},\{\phi_{j}(X)\}_1^{v}\right).
\end{align*}
Also, by lemma \ref{lemma:ub1} we have that for all $i=1,2,\ldots,v$
\begin{align*}
H(X'_i | \phi_{i}(X))=0 \implies H\left( X'_i | \{X'_j\}_{1}^{i-1},\{\phi_{j}(X)\}_1^v\right)=0.
\end{align*}
Using these in \eqref{eq:th2} we get that
\vspace{-1mm}\begin{align*}
& vm\log_2q + H\left(\{\phi_{i}(X)\}_{1}^{v}|\{X'_i\}_{1}^{v}\right) \leq vn\log_2q \\
\text{i.e.,}~ & vm \log_2q \leq vn\log_2q \implies \frac{m}{n} \leq 1.
\end{align*}
\end{proof}
\begin{theorem}
For a $2-(v,k,1)$ design $\mathcal{D}$, the coding capacity of the sum-network obtained using the above construction is at most $\frac{k(k-1)}{k(k-1)+v-1}$ if $ch(\mathcal{F}) \nmid (k-1)$, i.e., $ch(\mathcal{F})$ is not a divisor of $(k-1)$.
\label{thm:ubodd}
\end{theorem}
\begin{proof}
Let $X'_i=X_{p_i}+\sum_{B \in \langle p_i \rangle} X_B$ and $\phi_{i}(X)$ be the function transmitted on the $i$-th bottleneck.
Consider terminal $t_B$, for any $B \in \mathcal{B}$. Its incoming edges include $\{e_p : p \in B\}$. Then, from lemma \ref{lemma:ub1} it can evaluate
\begin{align}
\sum_{p_i \in B} X'_i=\sum_{p_i \in B}X_{p_i} + \sum_{p_i \in B}\sum_{B' \in \langle p_i \rangle}X_{B'}.
\label{eq:th1a}
\end{align}
Since $\mathcal{D}$ is a 2-design with $\lambda = 1$, for any $p, p' \in B$, we have that $\langle p \rangle \cap \langle p' \rangle=B$. By definition, $B \in \langle p \rangle~ \forall p \in B$ and $|B|=k$. These imply that the RHS in \eqref{eq:th1a} simplifies to
\begin{align}
\sum_{p_i \in B} X_{p_i}+ kX_B + \sum_{ B' \in \langle B \rangle \setminus B} X_{B'}
\label{eq:th1b}
\end{align} where $ \langle B \rangle \setminus B=\{B' \in \langle B \rangle: B' \neq B\}$.
From Lemma \ref{lemma:ub2}, $t_B$ can evaluate
\begin{align*}
\sum_{ p \in B} X_{p}  + \sum_{B' \in \langle B \rangle} X_{B'} = \sum_{ p_i \in B} X_{p_i}  + X_B + \sum_{B' \in \langle B \rangle \setminus B} X_{B'}
\end{align*} from the bottleneck edges.
This implies that terminal $t_B$ can compute the value of $(k-1)X_B$. 
Since $ch(\mathcal{F}) \nmid (k-1)$, $t_B$ can compute the value of $X_B$.

From lemma \ref{lemma:ub1}, $t_{p_i}$ can evaluate $X'_{i}$ for all $p_i \in P$. Consider a hypothetical terminal $t_\star$ which can observe the functions $\phi_{i}(X)$ for all $p_i \in P$. Since $t_B$ can find the value of $X_B$ from the bottleneck edges incident to it, so can $t_\star$. Similarly, $t_\star$ can find $X'_{i}$ for all $p_i \in P$ (since $t_{p_i}$ can find it). Then $t_\star$ is able to find $X_p$ and $X_B$ for any $p$ and $B$. 
Using a cut-based argument for $t_\star$, we get that
\begin{align*}
q^{(v+b)m} \leq q^{vn} \implies \frac{m}{n} \leq \frac{v}{v+b}.
\end{align*}
For $\mathcal{D}$, $b=\frac{v(v-1)}{k(k-1)}$. Substituting this above gives us the result.
\end{proof}

\vspace{-1mm}\section{Linear network code for constructed sum-network}
\label{sec:lincode}
We now describe $(m,n)$-fractional linear network codes whose rate matches the upper bound described above. 
\vspace{-1mm}\subsection{Case when $ch(\mathcal{F}) \mid (k-1)$}
In this case, we set $m=n=1$. We define the edge functions and terminal functions as follows ($i=1\hdots v; j=1 \hdots b$)
\begin{align*}
\phi_{i}(X) = X'_i = X_{p_i}+\sum_{B \in \langle p_i \rangle} X_B,
\end{align*} \vspace{-3mm}
\begin{align}
\label{eq:even1}
\psi\left(\phi_{\text{In}\left(t_{p_i}\right)}(X)\right) = \phi_{i}(X) + \sum_{ p \neq p_i} X_p + \sum_{ B \notin \langle p_i \rangle} X_B,\end{align}
\begin{align}
\label{eq:even2}
\psi\left(\phi_{\text{In}\left(t_{B_j}\right)}(X)\right) =\hspace{-1mm} \sum_{p_l \in B_j} \phi_{l}(X) +\hspace{-1mm} \sum_{p \notin B_j} X_p +\hspace{-2mm} \sum_{B \notin \langle B_j \rangle} X_B.
\end{align}
\vspace{-1mm}

All other edges $e$ in $G$ are such that $|\text{In}(e)|=1$ and we set $\phi_e(X)=\phi_{\text{In}(e)}(X)$. The following is proved in \ref{app:lemma3}.
\begin{lemma}
The network code described above enables all the terminals to compute the sum of sources $Z=\sum_{ p \in P} X_p + \sum_{B \in \mathcal{B}} X_B$.
\label{lemma:lincode1}
\end{lemma}
\subsection{Case when $ch(\mathcal{F}) \nmid (k-1)$}
Let $v'=v-x, ~\text{where}~ v \equiv x \Mod k, ~\text{and}~b'=rv'/k.$ We describe a $(m,n)$ fractional linear network code with $m=v', n=v'+b'=v'+v'(v-1)/k(k-1)$. Note that the ratio $m/n$ is then equal to the upper bound described in theorem \ref{thm:ubodd}.
In order to specify the linear network code for this case, we first ``color'' the incidence matrix $A$ of the underlying BIBD $\mathcal{D}$ to obtain a matrix $A_c$ of the same size as $A$. 
The coloring procedure assigns a successive natural number to each non-zero element present in every column of $A$. The numbers in each column start from $1$ and go up to $k$.

\begin{example}
For an adjacency matrix as defined in \eqref{eq:A}, the coloring procedure returns
\begin{align}
A_c=
\begin{bmatrix}
1 & 0 & 1 & 1 & 0 & 0 & 0\\2 & 0 & 0 & 0 & 1 & 0 & 1\\3 & 1 & 0 & 0 & 0 & 1 & 0\\0 & 2 & 0 & 2 & 0 & 0 & 2\\0 & 3 & 2 & 0 & 2 & 0 & 0\\0 & 0 & 3 & 0 & 0 & 2 & 3\\0 & 0 & 0 & 3 & 3 & 3 & 0
\end{bmatrix}.
\label{eq:Acolor}
\end{align}
\end{example}
We arrange the $r$ blocks incident to each $p_i$ in increasing order of their block indices and denote them as $(B'_{i1}, B'_{i2}, \ldots, B'_{ir})$. Note that any $B'_{\alpha \beta}$ corresponds to a unique $B_{\gamma} \in \mathcal{B}$ where
\begin{align}
\gamma=\min_{t \in \{1,2,\ldots,b\}} t ~~\text{such that}~\sum_{l=1}^{t}A(\alpha,l)=\beta.
\end{align}

We can then define a \textit{selector function} $U_{\alpha\beta}$ for any $\alpha \in \{1,2,\ldots,v\}$ and $\beta \in \{1,2,\ldots,r\}$ as follows.
\begin{align*}
U_{\alpha \beta}&=\begin{bmatrix} \mathbf{0}_{\frac{v'}{k} \times \frac{v'}{k}(d-1)} & I_{\frac{v'}{k}} & \mathbf{0}_{\frac{v'}{k} \times \frac{v'}{k}(k-d)}\end{bmatrix}X_{B'_{\alpha \beta}}, ~\text{where}\\B'_{\alpha \beta}&=B_\gamma ~\text{and}~ d=A_c(\alpha,\gamma), ~~\gamma \in \{1,2,\ldots,b\}.
\end{align*} Here $I_a$ denotes the identity matrix of dimension $a \times a$ and $\mathbf{0}_{a \times b}$ denotes a zero matrix of dimension $a \times b$.

We now define the edge functions for the linear network code. We let
\begin{align}
\phi_i(X)=\begin{bmatrix}X_i^{'T} & U^T_{i1} & U^T_{i2} & \ldots & U^T_{ir}\end{bmatrix}^T
\end{align}
where $X'_i=X_{p_i}+\sum_{B \in \langle p_i \rangle} X_B$. 
All other edges $e$ in $G$ are such that $|\text{In}(e)|=1$ and we set $\phi_e(X)=\phi_{\text{In}(e)}(X)$.
By definition it is clear that $\phi_{i}(X)$ is a function of $\phi_{\text{In}(e_i)}(X)$. Furthermore, it is easy to see that $\phi_{i}(X)$ is of dimension $n \times 1$. This is because $X'_i$ has dimension $v' \times 1$, and $U_{ij}$ for all $j$ is of dimension $\frac{v'}{k} \times 1$, so that $X'_i$ has dimension $v'+r\left( \frac{v'}{k}\right)= v' + b' = n$. 

The terminal functions used in this network code are defined as follows. For any $p_i \in P$, let
\begin{align*}
Z_1=\sum_{p \neq p_i}X_p,~  Z_2=\sum_{B \notin \langle p_i \rangle}X_B
\end{align*}
such that $t_{p_i}$ can evaluate $Z_1$ and $Z_2$ from its direct edges.
We define the terminal functions for $t_{p_i}, p_i \in P$ as
\begin{align*}
\psi\left(\phi_{\text{In}\left(t_{p_i}\right)}(X)\right)&=\begin{bmatrix} I_{v'} &\mathbf{0}_{v' \times b'}\end{bmatrix} \phi_{i}(X)+Z_1 +Z_2\\&=X'_i+Z_1 +Z_2 =Z.
\end{align*}
\vspace{-2mm}
For terminals of the form $t_{B_j}, B_j \in \mathcal{B}$, we first let
\begin{align*}
Z_1=\sum_{p \notin B_j}X_p, ~ Z_2=\sum_{B \notin \langle B_j \rangle}X_B
\end{align*} such that both $Z_1$ and $Z_2$ are available to $t_{B_j}$ via direct edges.
The terminal functions for $t_{B_j}, B_j \in \mathcal{B}$ are evaluated as
\begin{align}
\psi\left(\phi_{\text{In}\left(t_{B_j}\right)}(X)\right)=&\hspace{-0.5mm}\sum_{p_i \in B_j}\begin{bmatrix} I_{v'} & \mathbf{0}_{v' \times b'}\end{bmatrix} \phi_{i}(X)+Z_1+Z_2 \nonumber \\
&-(k-1)\hspace{-0.5mm}\begin{bmatrix} U^T_{\alpha_1\beta_1} & U^T_{\alpha_2\beta_2} & \hspace{-1mm}\ldots\hspace{-1mm} & U^T_{\alpha_k\beta_k}\end{bmatrix}^T
\label{eq:psi}
\end{align}
where $U_{\alpha\beta}$ is the selector function and 
\begin{align*}
A_c(\alpha_u,j)=u,~ B'_{\alpha_u\beta_u}=B_j ~\text{for all}~ u \in \{1,2,\ldots,k\}.
\end{align*}
From definition it can be seen that all the elements in the RHS of \eqref{eq:psi} are contained in the set In($t_{B_j}$). In addition we have the following lemma.
\begin{lemma} \label{lemma:Xb}
For all $B_j \in \mathcal{B}$, we have that
\begin{align*}
X_{B_j}=\begin{bmatrix} U^T_{\alpha_1\beta_1} & U^T_{\alpha_2\beta_2} & \ldots & U^T_{\alpha_k\beta_k}\end{bmatrix}^T
\end{align*}
\end{lemma}
\begin{proof}
It can be seen that $U_{\alpha_u\beta_u}$ consists of components of the source $X_{B_j}$. In addition, we have that
\begin{align*}
\begin{bmatrix} U_{\alpha_1\beta_1} \\ U_{\alpha_2\beta_2} \\ \vdots \\ U_{\alpha_k\beta_k}\end{bmatrix}=
\begin{bmatrix}
I_{\frac{v'}{k}} & \hdots & \mathbf{0} & \mathbf{0}\\
\mathbf{0} & I_{\frac{v'}{k}} &  & \mathbf{0}\\
\vdots &  & \ddots & \vdots\\
\mathbf{0} & \mathbf{0} & \hdots & I_{\frac{v'}{k}}
\end{bmatrix}X_{B_j} = X_{B_j}.
\end{align*}
\end{proof}
The proof of the next lemma appears in \ref{app:theorem3}.
\begin{lemma}
All terminals $t_{B_j}, B_j \in \mathcal{B}$ can evaluate the sum of sources, i.e. $H\left(Z|\psi\left(\phi_{\text{In}\left(t_{B_j}\right)}(X)\right)\right)=0$.
\label{lemma:lincode2}
\end{lemma}
\begin{example}
We describe the linear network code for the sum-network described in example \ref{ex:fanograph}. Here $k=3$ and we describe a linear network code for the case when $ch(\mathcal{F}) \nmid (k-1)=2$, i.e., $\mathcal{F}$ has odd characteristic. Then $v'=b'=6$. The linear network code has rate $m/n=6/12=1/2$. To describe the network code we first obtain a coloring of the adjacency matrix which is described in \eqref{eq:Acolor}. 

Suppose  $X$ is a $6 \times 1$ vector. Let
\begin{align*} X(a:b)^T=\begin{bmatrix}\mathbf{0}_{(b-a+1) \times (a-1)} & I_{(b-a+1)} & \mathbf{0}_{(b-a+1) \times (6-b)}\end{bmatrix}X. \end{align*}  Then the edge vectors are as follows:
\begin{align*}
\phi_{1}(X)&=\begin{bmatrix} X^{'T}_1 & X_A(1:2) & X_C(1:2) & X_D(1:2)\end{bmatrix}^T,\\
\phi_{2}(X)&=\begin{bmatrix} X^{'T}_2 & X_A(3:4) & X_E(1:2) & X_G(1:2)\end{bmatrix}^T,\\
\phi_{3}(X)&=\begin{bmatrix} X^{'T}_3 & X_A(5:6) & X_B(1:2) & X_F(1:2)\end{bmatrix}^T,\\
\phi_{4}(X)&=\begin{bmatrix} X^{'T}_4 & X_B(3:4) & X_D(3:4) & X_G(3:4)\end{bmatrix}^T,\\
\phi_{5}(X)&=\begin{bmatrix} X^{'T}_5 & X_B(5:6) & X_C(3:4) & X_E(3:4)\end{bmatrix}^T,\\
\phi_{6}(X)&=\begin{bmatrix} X^{'T}_6 & X_C(5:6) & X_F(3:4) & X_G(5:6)\end{bmatrix}^T.\\
\phi_{7}(X)&=\begin{bmatrix} X^{'T}_7 & X_D(5:6) & X_E(5:6) & X_D(5:6)\end{bmatrix}^T.
\end{align*}
We now verify that $\psi\left(\phi_{\text{In}(t_C)}(X)\right)=Z$. The same exercise can be repeated for other terminals. For $t_C, ~B_j=C$ and
\begin{align*}
A_c(\alpha_1,3)=1 \implies \alpha_1=1, ~B'_{1\beta_u}=C \implies \beta_u=2.
\end{align*}
Similarly, we can compute that $\alpha_2=5, \beta_2=2, \alpha_3=6,\beta_3=1$. This implies that
\begin{align*}
\begin{bmatrix}
U_{\alpha_1\beta_1}\\U_{\alpha_2\beta_2}\\U_{\alpha_3\beta_3}
\end{bmatrix}
=
\begin{bmatrix}
U_{12}\\U_{52}\\U_{61}
\end{bmatrix}
=
\begin{bmatrix}
X_C(1:2)^T\\X_C(3:4)^T\\X_C(5:6)^T
\end{bmatrix}=X_C.
\end{align*}
It can then be seen that $t_C$ can compute the sum of all sources.
\end{example}

\subsection{Sum-networks from Steiner triple systems}
A Steiner triple system (STS) of order $v$ is a $2-(v,k,\lambda)$ design with $k=3$ and $\lambda =1$. The Fano plane is an example of a STS. It is well known that there exist Steiner triple systems for all orders $v \equiv 1 ~\text{or}~3 \Mod 6$ \cite{Stinson}.
\begin{theorem}
Sum-networks constructed from STS with $v$ points are such that their capacity equals $1$ if $ch(\mathcal{F}) = 2$ and $\frac{6}{5+v}$ is $ch(\mathcal{F})$ is odd.
\end{theorem}
\begin{proof}
For a sum-network constructed using a STS, applying Theorem \ref{thm:ub1}, Lemma \ref{lemma:lincode1}, Theorem \ref{thm:ubodd} and Lemma \ref{lemma:lincode2}, we have the desired result. 
Thus, the message alphabet can greatly impact the coding capacity of a sum-network. 
\end{proof}

\section{Conclusion}
We have constructed a family of sum-networks whose network coding capacity is dependent on the message alphabet (specifically, the characteristic of the finite field) chosen for communication. Previous results in this line described only the solvability of sum-networks for different alphabets using linear network coding. It was not clear how the coding capacity of a sum-network would be affected due to choice of message alphabet, if at all. The results in this paper address that issue.


%
\bibliographystyle{IEEETran}
\bibliography{tip,career_bib}

\section{Appendix}
\subsection{Proof of Lemma \ref{lemma:ub1}}
\label{app:lemma1}
We let for any $p_i \in P$
\begin{align*}
Z_1 = \sum_{p \neq p_i} X_p,~
Z_2 = \sum_{B \in \langle p_i \rangle} X_B ~\text{and}~
Z_3 = \sum_{B \notin \langle p_i \rangle} X_B
\end{align*}
such that $Z= X_{p_i}+ Z_1 + Z_2 + Z_3$.
By our assumption, we know that $Z$ can be evaluated from $\phi_{\text{In}\left(t_{p_i}\right)}(X)$ for all $i=1,2, \ldots,v$, i.e.,
$ H\left( Z | \phi_{\text{In}(t_{p_i})}(X)\right)=0$.

Since $\{X_p : p \neq p_i\}$ and $\{X_B : B \notin \langle p_i \rangle\}$ determine the value of $Z_1$ and $Z_3$ respectively and are also a subset of $\phi_{\text{In}\left(t_{p_i}\right)}(X)$, we get that
\begin{align*}
 &H\left( Z | \phi_{\text{In}\left(t_{p_i}\right)}(X)\right)=0 \implies \hspace{-1mm} H\left( X_{p_i} +Z_2 | \phi_{\text{In}\left(t_{p_i}\right)}(X)\right)=0, \\&\text{i.e.,}~H\left( X_{p_i} \hspace{-1mm}+ Z_2 | \phi_{i}(X),\hspace{-0.5mm}\{X_p : p \neq p_i\},\hspace{-0.5mm}\{X_B : B \notin \langle p_i \rangle\}\right)=0.
\end{align*}
Since all the sources are i.i.d., it can be verified that $X_{p_i}+Z_2$ is conditionally independent of both $\{X_p : p \neq p_i\}$ and $\{X_B : B \notin \langle p_i \rangle\}$ given $\phi_{i}(X)$. Hence, we have that $H(X_{p_i}+Z_2 | \phi_{i}(X))=0$.

\subsection{Proof of Lemma \ref{lemma:ub2}}
\label{app:lemma2}
We let for any $B_j \in \mathcal{B}$
\begin{align*}
Z_1=\sum_{p \in B_j}\hspace{-1mm} X_p, Z_2=\sum_{p \notin B_j}\hspace{-1mm} X_p, Z_3=\sum_{B \in \langle B_j \rangle}\hspace{-1mm}X_B, Z_4=\sum_{B \notin \langle B_j \rangle}\hspace{-1mm}X_B
\end{align*}
such that $Z= Z_1+Z_2+Z_3+Z_4$.
By our assumption, for all $i=1,2,\ldots,b$, $H\left(Z | \phi_{\text{In}\left(t_{B_j}\right)}(X)\right)=0$.

The sets $ \{X_p : p \notin B_j\}$ and $ \{X_B : B \notin \langle B_j \rangle\}$ determine the value of $Z_2$ and $Z_4$ respectively and are also subsets of $\phi_{\text{In}\left(t_{B_j}\right)}(X)$.  Hence, we have that
\begin{align*}
H\left( Z | \phi_{\text{In}\left(t_{B_j}\right)}(X)\right)=&0 \implies H\left( Z_1 + Z_3 |\phi_{\text{In}\left( t_{B_j}\right)}(X) \right)=0,\\ \text{where,}~
\phi_{\text{In}\left( t_{B_j}\right)}(X)=&\{\phi_{i}(X) : p_i \in B_j\} \cup \{X_{p} : p \notin B_j\} \\& \cup \{X_B : B \notin \langle B_j \rangle\}.
\end{align*}
Since all the sources are i.i.d., it can be verified that $Z_1+Z_3$ is conditionally independent of both $\{X_p : p \notin B_j\}$ and $\{X_B : B \notin \langle B_j \rangle\}$ given the set of random variables $\{\phi_{i}(X) : p_i \in B_j\}$. This gives us the result that $H(Z_1+Z_3 | \{\phi_{i}(X) : p_i \in B_j\})=0$.

\subsection{Proof of claim in Theorem \ref{thm:ub1}}
\label{app:claim}
We want to prove the identity
\begin{align}
P(X_1'=x_1',\ldots,X_v'=x_v')=\prod_{i=1}^{v}P(X_i'=x_i').
\label{eq:x'}
\end{align}
Let $\mathcal{X}$ be the set of solutions to the following system of linear equations:
\begin{align*}
x_{p_1}+&\sum_{B \in \langle p_1 \rangle} x_B=x'_1\\
x_{p_2}+&\sum_{B \in \langle p_2 \rangle} x_B=x'_2\\
&\hspace{4mm}\vdots\\
x_{p_v}+&\sum_{B \in \langle p_v \rangle} x_B=x'_v.
\end{align*}Choosing any value from $\mathcal{F}^m$ for each of the variables in $\{x_B : B \in \mathcal{B}\}$ fixes the value of $x_{p_i}$ given $x'_i$ for all $i \in \{1,2,\ldots,v\}$ and vice versa.
The probability of the set in \eqref{eq:x'} is equal to the probability with which the set of random variables $\left\{\{X_{p_i}\}_1^v \cup \{X_{B_j}\}_1^b\right\}$ take values in $\mathcal{X}$. Since all random variables in $\left\{\{X_{p_i}\}_1^v \cup \{X_{B_j}\}_1^b\right\}$ are uniform i.i.d. over $\mathcal{F}^m$ with $|\mathcal{F}|=q$ we can expand the LHS of \eqref{eq:x'} as
\begin{align*}
&\sum_{\mathcal{X}} \prod_{i=1}^{v}P(X_{p_i}=x_{p_i})\prod_{j=1}^{b}P(X_{B_j}=x_{B_j})\\
&=\left[\frac{1}{q^m}\right]^b\left[\frac{1}{q^m}\right]^v\left(q^m\right)^b=\left[\frac{1}{q^m}\right]^v.
\end{align*}
Also, we have that
\begin{align*}
P(X'_i=x'_i)=\sum_{\mathcal{X}'}P(X_{p_i}=x_{p_i})\prod_{B\in \langle p_i \rangle}P(X_{B}=x_{B})
\end{align*} where $\mathcal{X}'$ is the set of solutions to the variables $\{x_{p_i},\{x_B\}_{B \in \langle p_i \rangle}\}$ such that $x'_i=x_{p_i}+\sum_{B \in \langle p_i \rangle}x_B$.
Then
\begin{align*}
P(X'_i=x'_i)=\frac{1}{q^m}\left[\frac{1}{q^m}\right]^r\left(q^m\right)^r=\frac{1}{q^m},
\end{align*}
and the RHS of \eqref{eq:x'} is
\begin{align*}
\prod_{i=1}^{v}P(X_i'=x_i')=\left[\frac{1}{q^m}\right]^v.
\end{align*}
\subsection{Proof of Lemma \ref{lemma:lincode1}}
\label{app:lemma3}
 Consider $p_i \in P$, and let
\begin{align*}
Z_1=\sum_{p \neq p_i}X_p,~Z_2=\sum_{B \in \langle p_i \rangle}X_B~\text{and}~ Z_3=\sum_{B \notin \langle p_i \rangle}X_B
\end{align*}
such that $\phi_{i}(X)=X_{p_i}+Z_2$. Then from \eqref{eq:even1} for $i=1,2,\hdots v$,
\begin{align*}
\psi\left(\phi_{\text{In}\left(t_{p_i}\right)}(X)\right)\hspace{-0.5mm}=\phi_{i}(X)\hspace{-0.5mm}+\hspace{-0.5mm}Z_1\hspace{-0.5mm}+\hspace{-0.5mm}Z_3\hspace{-0.5mm}= \hspace{-0.5mm}X_{p_i}\hspace{-0.5mm}+Z_2\hspace{-0.5mm}+\hspace{-0.5mm}Z_1\hspace{-0.5mm} +\hspace{-0.5mm}Z_3\hspace{-0.5mm}=\hspace{-0.5mm}Z.
\end{align*}
Now we look at terminals of the form $t_{B_j}, B_j \in \mathcal{B}$. For this, we let
for any $B_j \in \mathcal{B}$
\begin{align*}
&Z_1 = \sum_{p \in B_j} X_p,~
Z_2 = \sum_{p \notin B_j} X_p, \\
&Z_3 = \sum_{B \in \langle B_j \rangle \setminus B_j} X_B,~\text{and}~
Z_4 = \sum_{B \notin \langle B_j \rangle} X_B
\end{align*} such that $Z=X_{B_j}+Z_1+Z_2+Z_3+Z_4$.
 Since $\mathcal{D}$ is a 2-design with $\lambda = 1$, we have that for any $p,p' \in B$, $\langle p \rangle \cap \langle p'\rangle=B$. This implies that
\begin{align*}
\sum_{ p_l \in B_j} \phi_{l}(X)=Z_1 + Z_3 + kX_{B_j}.
\end{align*}
Since $ch(\mathcal{F}) \mid (k-1)$, we have that $kX_{B_j} = X_{B_j}$. Using this in \eqref{eq:even2} we get that for $j=1,2, \hdots b$
\begin{align*}
\psi\left(\phi_{\text{In}\left(t_{B_j}\right)}(X)\right)=Z_1 + Z_3 + X_{B_j}+ Z_2+ Z_4=Z.
\end{align*}
\subsection{Proof of Lemma \ref{lemma:lincode2}}
\label{app:theorem3}
We let
\begin{align*}
&Z_1 = \sum_{p \notin B_j} X_p,~
Z_2 = \sum_{B \notin \langle B_j \rangle} X_B, \\
&Z_3 = \sum_{B \in \langle B_j \rangle \setminus B_j} X_B,~\text{and}~
Z_4 = \sum_{p \in B_j} X_p
\end{align*} such that $Z=X_{B_j}+Z_1+Z_2+Z_3+Z_4$.
Then using \eqref{eq:th1b},
\begin{align*}
&\sum_{ p_i \in B_j}\begin{bmatrix} I_{v'} & \mathbf{0}_{v' \times b}\end{bmatrix} \phi_{i}(X)=\sum_{ p_i \in B_j} X'_i=Z_3+Z_4+kX_{B_j}.
\end{align*}
Using this and lemma \ref{lemma:Xb} in \eqref{eq:psi} we get that
\begin{align*}
\psi(\phi_{\text{In}(t_{B_j})}(X))\hspace{-0.5mm}=\hspace{-0.5mm}Z_3\hspace{-0.5mm}+\hspace{-0.5mm}Z_4\hspace{-0.5mm}+\hspace{-0.5mm}kX_{B_j}\hspace{-1mm}-\hspace{-0.5mm}(k-1)X_{B_j}\hspace{-1mm}+\hspace{-0.5mm} Z_1\hspace{-0.5mm}+\hspace{-0.5mm}Z_2\hspace{-0.5mm}=\hspace{-0.5mm}Z.
\end{align*}
\subsection{Remark about non-applicability of Theorem \textit{VI.5} in \cite{CannonsDFZ06} for sum-networks}
Theorem VI.5 in\cite{CannonsDFZ06} states that the message capacity of a network is independent of the alphabet used.
Consider a simple sum-network shown in Figure \ref{fig:sumnw_appendix}, terminal $t$ wants to evaluate $X_1+X_2$ where $X_1,X_2 \in A$ are random variables observed at source nodes $s_1,s_2$ respectively.
\begin{figure}
\centering
\includegraphics[]{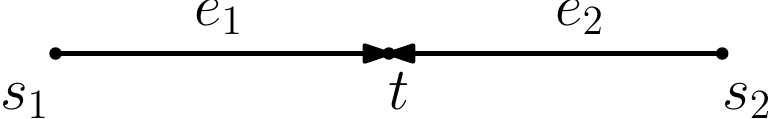}
\caption{A simple sum-network. Both edges can transmit one symbol in $A$ from tail to head in one channel use.}
\label{fig:sumnw_appendix}
\end{figure}
We have a simple scalar network code (rate $k/n=1, k=n=1$) that satisfies the problem, described as follows.
\begin{enumerate}
\item Edge functions:
\begin{align*}
\phi_{e_1}(X_1) \equiv f_{e_1}(X_1)&=X_1,\\\phi_{e_2}(X_2) \equiv f_{e_2}(X_2)&=X_2.
\end{align*}
\item Decoding function:
\begin{align*}
\psi(\phi_{e_1}(X_1),\phi_{e_2}(X_2)) \equiv &f_{t,\mathbf{\Sigma}}\left(f_{e_1}(X_1),f_{e_2}(X_2)\right)\\=&f_{e_1}(X_1)+f_{e_2}(X_2)
\end{align*} where $\mathbf{\Sigma}=X_1+X_2$ is the only message terminal $t$ is interested in.
\end{enumerate}
We use the procedure outlined in \cite{CannonsDFZ06} to extend the network code for another alphabet $B$. Let $A=GF(3), B=GF(2)$. Setting $\epsilon = 2^{1-\gamma}/\log_2 3$ where $\gamma>1$, we obtain the following values
\begin{align*}
t=2^\gamma, n'=\left\lceil \frac{2^\gamma}{\log_2 3} \right\rceil ~\text{and}~ k'=\lfloor n' \rfloor -1.
\end{align*}
Let $h_0 : B \rightarrow A$ be such that
\begin{align*}
h_0(x)=\left\{\begin{matrix}
0 & ~\text{if}~x=0,\\
1 & ~\text{if}~x=1.
\end{matrix}\right.
\end{align*} and let $\hat{h}_0 : A \rightarrow B$ such that $\hat{h}_0(h_0(x))=x ~\text{for all}~ x \in B$ and arbitrary otherwise.
Then we can define an injection $\mathbf{h_0} : B^{k'} \rightarrow A^{t}$ as the componentwise application of $h_0$ to each of the elements in the argument. That is
\begin{align*}
\mathbf{h_0}(b_1,b_2,\ldots,b_{k'})= \begin{bmatrix} h_0(b_1) & h_0(b_2) & \ldots & h_0(b_{k'}) & \mathbf{0}_{t-k'} \end{bmatrix}
\end{align*} where $b_1,b_2, \ldots b_{k'} \in B$ and $\mathbf{0}_{t-k'}$ is a zero vector with $t-k'$ components. We define $\mathbf{\hat{h}_0} : A^t \rightarrow B^{k'}$ as
\begin{align*}
\mathbf{\hat{h}_0}(a_1,a_2,\ldots,a_t) = \begin{bmatrix} \hat{h}_0(a_1) & \hat{h}_0(a_2) & \ldots & \hat{h}_0(a_{k'}) \end{bmatrix}
\end{align*} where $a_1,a_2,\ldots,a_t \in A$.

Also we let $\mathbf{h}: A^t \rightarrow B^{n'}$ be an arbitrary injection and $\mathbf{\hat{h}}: B^{n'} \rightarrow A^t$ is such that $\mathbf{\hat{h}}(\mathbf{h}(\mathbf{x}))=\mathbf{x} ~\text{for all}~ \mathbf{x} \in A^t$ and arbitrary otherwise.
We now use the extended network code to satisfy the sum network for when the source random variables take values in the alphabet $B^{k'}$. Suppose a particular realization of $\mathbf{X_1} \in B^{k'}$ and $\mathbf{X_2} \in B^{k'}$ is such that
\begin{align*}
\mathbf{x_1}=(1,1,\ldots,1)=\mathbf{1}_{k'} ~\text{and}~ \mathbf{x_2}=(1,1,\ldots,1)=\mathbf{1}_{k'}.
\end{align*}
Following steps in \cite{CannonsDFZ06} for the decoding function we get that
\begin{align*}
g_{t,\mathbf{\Sigma}}(f_{e_1}(\mathbf{x_1}),f_{e_2}(\mathbf{x_2}))&=\mathbf{\hat{h}_0}(f_{t,\mathbf{\Sigma}}(f_{e_1}(\mathbf{h_0}(\mathbf{x_1})),f_{e_2}(\mathbf{h_0}(\mathbf{x_2}))))\\
&=\mathbf{\hat{h}_0}(\mathbf{h_0}(\mathbf{x_1})+\mathbf{h_0}(\mathbf{x_2}))\\
&=\mathbf{\hat{h}_0}([\mathbf{1}_{k'} ~\mathbf{0}_{t-k'}]+[\mathbf{1}_{k'} ~\mathbf{0}_{t-k'}])\\
&=\mathbf{\hat{h}_0}([\mathbf{2}_{k'} ~\mathbf{0}_{t-k'}])
\end{align*} where $\mathbf{2}_{k'}$ is a vector of $k'$ $2$'s.

Since $\hat{h}_0(2)$ is arbitrarily assigned, $\mathbf{\hat{h}_0}([\mathbf{2}_{k'} ~\mathbf{0}_{t-k'}])$ need not equal $\mathbf{0}_{k'}$ which is the right value of $\mathbf{X_1}+\mathbf{X_2}$. Thus the extension of the network code does not correctly evaluate the sum in $B^{k'}$.

The characteristics of a particular alphabet only affect the \textit{value} of a function of the source random variables and not the random variables themselves. The extension of the network code from one alphabet to another works for the case of multiple unicast as the messages demanded by any terminal are a subset of all the messages observed in the network and not a function of them.

\end{document}